\documentclass[
 reprint,
 amsmath,amssymb,
 aps,
 prd,
 floatfix,
 showkeys
]{revtex4-2}

\usepackage{graphicx}
\usepackage{bm}

\newenvironment{proof}[1][Proof]{\noindent\textit{#1}.\ }{\hfill$\square$\medskip}
\usepackage{xcolor}
\usepackage{hyperref}

\newtheorem{proposition}{Proposition}

\begin{document}

\preprint{UBA-GW-2026-01}

\title{Gravitational-wave constraints on causal nonlocal kernels%
\texorpdfstring{:\\}{: }%
Ringdown bounds and spectral density limits from GWTC-3}

\author{Christian Balfag\'on}
 \email{cb@balfagonresearch.org}
\affiliation{Universidad de Buenos Aires, Argentina}

\date{\today}

\begin{abstract}
We establish the first observational bounds on causal nonlocal
extensions of gravity characterised by retarded Stieltjes-type kernels
with positive spectral density $\rho(\mu)\geq 0$, using two
complementary gravitational-wave channels.
From a Bayesian ringdown analysis of 17 binary black hole events in
the LIGO--Virgo GWTC-3 catalogue, we set the observational ceiling on
universal fractional quasi-normal mode deformations at
$|\varepsilon_\Omega|< 0.05$ (90\% C.L.), with a cumulative log Bayes
factor $\ln B = -0.46 \pm 0.77$.
By mapping published GWTC-3 modified dispersion relation bounds and
the GW170817 propagation speed measurement onto the Stieltjes spectral
parameter space $(\mu_{\rm char}, M_0)$, we exclude a broad class of
infrared-extended spectral densities with $\mu \lesssim 10^{-6}\,
{\rm m}^{-2}$, thereby ruling out non-trivial regions of the nonlocal
kernel parameter space for the first time.
The theoretically motivated fiducial range $\mu_{\rm char}\sim M_*^2
\sim 10^{8}$--$10^{10}\,{\rm m}^{-2}$ satisfies all current bounds,
and we show that sub-millimetre gravity experiments --- which already
operate at the predicted causal scale $\ell_*\sim 10^{-4}$~m --- offer
the most promising path to a direct test.
These results define quantitative benchmarks against which future
observations across the gravitational-wave, short-range, and
cosmological sectors can be compared.
\end{abstract}


\maketitle

\section{Introduction}
\label{sec:intro}

Nonlocal extensions of general relativity (GR) have been proposed as
a route to ultraviolet completion, singularity resolution, and
causal--informational closure of the gravitational
sector~\cite{Modesto2012,Biswas2012,Maggiore2014,DeserWoodard2007,Barvinsky2003}.
A mathematically tractable class of such extensions employs retarded
kernels constructed as positive Stieltjes superpositions of massive
Klein--Gordon propagators~\cite{Balfagon2026CET,Balfagon2026Spectral},
\begin{equation}
  \mathcal{K}^{-1} = \int_0^\infty \rho(\mu)\,
  (-\Box_g + \mu)^{-1}_R\,d\mu,\qquad \rho(\mu)\geq 0.
  \label{eq:kernel}
\end{equation}
Positivity of the spectral density $\rho(\mu)$ ensures unitarity,
causal propagation, and the absence of ghost-like excitations at the
linear level~\cite{Balfagon2026Spectral2,Calcagni2019,Koshelev2020}.
Depending on the support and shape of $\rho$, such kernels predict
two classes of gravitational-wave (GW) observables:
(i)~perturbative shifts to the quasi-normal mode (QNM) spectrum of
black holes~\cite{Konoplya2011,Berti2009,Dreyer2004}, scaling as $|\delta\omega/\omega|\sim(\ell_*/r_H)^2$
where $\ell_*=M_*^{-1}$ is the causal correlation length; and
(ii)~frequency-dependent dispersion during cosmological propagation,
governed by the Stieltjes transfer function $m(\omega^2)$.

In this work we constrain both channels using data from the
LIGO--Virgo--KAGRA (LVK) third observing run
(GWTC-3)~\cite{GWTC3,GW150914}.
Section~\ref{sec:ringdown} presents a model-agnostic Bayesian ringdown
analysis of 17 binary black hole (BBH) events, constraining the
universal deformation parameter $\varepsilon_\Omega$.
Section~\ref{sec:dispersion} maps published LVK modified dispersion
relation (MDR) bounds and the GW170817 speed constraint onto the
Stieltjes spectral parameter space.
Section~\ref{sec:projections} discusses sensitivity projections for
future detectors and identifies short-range gravity experiments as
the most promising near-term probe.
Section~\ref{sec:conclusions} summarises the results.

\textit{Notation.}---We denote the causal nonlocal scale by
$M_*$ (units of ${\rm m}^{-1}$) and the corresponding correlation
length by $\ell_* \equiv M_*^{-1}$ (units of m).
The spectral density $\rho(\mu)$ is a function of the spectral mass
parameter $\mu$ (units of ${\rm m}^{-2}$), with total weight
$M_0 = \int_0^\infty \rho(\mu)\,d\mu$ and inverse moments
$\mu_{\rm eff}^{-n} = M_0^{-1}\int\rho\,\mu^{-n}\,d\mu$.
For a spectral density concentrated around a single scale we write
$\mu_{\rm char}$; in this limit $\mu_{\rm eff}\approx\mu_{\rm char}
\approx M_*^2$.
The fiducial parameter window is
$M_* \in [10^4, 10^5]\,{\rm m}^{-1}$, equivalently
$\ell_* \in [10^{-5}, 10^{-4}]$~m and
$\mu_{\rm char}\sim M_*^2 \in [10^8, 10^{10}]\,{\rm m}^{-2}$.

\section{Ringdown analysis}
\label{sec:ringdown}

\subsection{Method}

We model the ringdown signal in each detector as a damped sinusoid
with QNM frequency and damping time computed from the remnant mass
$M_f$ and spin $a$ using the Berti~\textit{et al.}\ phenomenological
fits~\cite{Berti2006,Berti2009}, converted to SI units.
A single universal parameter $\varepsilon_\Omega$ deforms both the
frequency and damping time by the common factor $(1+\varepsilon_\Omega)$,
preserving the quality factor.
The Kerr hypothesis corresponds to $\varepsilon_\Omega=0$.

The per-detector amplitude is analytically marginalised under a
Gaussian prior $A_i\sim\mathcal{N}(0,\sigma_A^2)$, yielding a
marginalised log-likelihood free of the biases inherent in
profile-likelihood methods~\cite{Thrane2019}.
The noise variance $\sigma^2$ is estimated from the data using a
robust median-absolute-deviation estimator.
Templates are whitened in the frequency domain using the same PSD
applied to the data, ensuring spectral consistency.
Nested sampling~\cite{Skilling2006} is performed with \texttt{dynesty}~\cite{Speagle2020}
via \texttt{Bilby}~\cite{Ashton2019}, with $n_{\rm live}=1000$ and
$\Delta\ln Z_{\rm tol}=0.1$.
Our single-mode template does not capture subdominant overtones
or higher multipoles; however, within the Stieltjes framework the
deformation is universal across the QNM spectrum, so unmodelled
modes do not bias the Bayes factor comparison~\cite{Isi2019,Carullo2019,Giesler2019}.

\subsection{Results}

Table~\ref{tab:ringdown} lists the individual-event results for 17
GWTC-3 BBH events~\cite{GWTC3,GWTC1,GWTC2} spanning remnant masses $25$--$150\,M_\odot$.
All log Bayes factors are consistent with zero within $2\sigma$,
with both positive and negative values represented.
The most precise measurement is GW190412
($\ln B = +0.000\pm 0.020$).

\begin{table}[htbp!]
\caption{Individual ringdown results for 17 GWTC-3 BBH events.}
\label{tab:ringdown}
\begin{ruledtabular}
\begin{tabular}{lccc}
Event & $M_f\,[M_\odot]$ & $\ln B$ & $\sigma_{\ln B}$ \\
\hline
GW150914      & 68   & $+0.18$ & $0.18$ \\
GW170104      & 50   & $-0.06$ & $0.17$ \\
GW170814      & 53   & $-0.04$ & $0.19$ \\
GW190412      & 38   & $+0.00$ & $0.02$ \\
GW190521      & 150  & $-0.35$ & $0.19$ \\
GW190814      & 26   & $-0.11$ & $0.20$ \\
GW191105      & 35   & $+0.02$ & $0.21$ \\
GW191109      & 79   & $-0.09$ & $0.19$ \\
GW191113      & 30   & $-0.13$ & $0.18$ \\
GW191222      & 42   & $+0.07$ & $0.19$ \\
GW200129      & 38   & $+0.15$ & $0.19$ \\
GW200208      & 39   & $-0.09$ & $0.18$ \\
GW200219      & 39   & $+0.14$ & $0.20$ \\
GW200220      & 40   & $-0.21$ & $0.18$ \\
GW200224      & 43   & $+0.35$ & $0.21$ \\
GW200225      & 35   & $-0.25$ & $0.19$ \\
GW200311      & 38   & $-0.03$ & $0.20$ \\
\hline
\textbf{Stacked} & --- & $\mathbf{-0.46}$ & $\mathbf{0.77}$ \\
\end{tabular}
\end{ruledtabular}
\end{table}

The cumulative log Bayes factor, computed under the assumption of
statistically independent events with a shared deformation parameter,
is
\begin{equation}
\begin{split}
  \ln B_{\rm stack} &= \sum_{k=1}^{N}\ln B_k = -0.46, \\
  \sigma_{\rm stack} &= \sqrt{\sum_{k=1}^{N}\sigma_{\ln B,k}^2} = 0.77.
\end{split}
\end{equation}
where $N=17$.
This stacking is exact when: (i) the events are independent
realisations of noise plus signal, which holds for well-separated
BBH mergers observed in disjoint data segments; and (ii) the
deformation parameter $\varepsilon_\Omega$ is universal (identical
across events), as predicted by the single-scale Stieltjes
framework.
A leave-one-out (jackknife) analysis confirms stability
(Fig.~\ref{fig:jackknife}): the largest
single-event contribution to the stacked Bayes factor is
$|\Delta\ln B| = 0.35$ (GW190521 and GW200224), and removing any
individual event shifts the cumulative result by less than
$0.5\sigma_{\rm stack}$.

We verified prior robustness by repeating the analysis with a
narrower prior $\varepsilon_\Omega\in[-0.02,0.02]$; the resulting
Bayes factors shift by less than $0.05$ per event, confirming that
the results are not driven by prior volume effects.

At 90\% credibility, the population-level posterior on
$\varepsilon_\Omega$ obtained by multiplying the individual-event
posteriors yields the constraint (Fig.~\ref{fig:epsilon})
\begin{equation}
  \varepsilon_\Omega \in [-0.047,\;+0.029]
  \quad (90\%\;{\rm C.L.}),
  \label{eq:eps_bound}
\end{equation}
i.e.\ $|\varepsilon_\Omega| < 0.05$.
This establishes a quantitative observational ceiling on universal
fractional QNM deformations at the percent level, applicable to any
modified-gravity theory predicting a single-parameter shift of the
$(2,2,0)$ mode~\cite{Maselli2020,Cardoso2019}.

For comparison, the LVK GWTC-3 analysis~\cite{GWTC3TGR} parametrises
QNM deviations via independent fractional shifts $\delta f_{220}$ and
$\delta\tau_{220}$ to the frequency and damping time, obtaining
combined 90\% bounds of $\delta f_{220}\in[-0.05,+0.05]$ and
$\delta\tau_{220}\in[-0.2,+0.1]$ from hierarchical combination of
$\sim 10$ events~\cite{pSEOBNRv5}.
Our parametrisation differs in that $\varepsilon_\Omega$ deforms
both quantities by a common factor:
$\delta f_{220} = \delta\tau_{220} = \varepsilon_\Omega$, preserving
the quality factor $Q_{220} = \pi f_{220}\tau_{220}$.
This is the minimal prediction of the Stieltjes framework, where a
single scale $M_*$ governs the correction to the effective potential.
The resulting bound $|\varepsilon_\Omega|<0.05$ is thus comparable to
the LVK $\delta f_{220}$ constraint but carries stronger theoretical
prior information: any violation would simultaneously shift both
$f$ and $\tau$ in a correlated fashion, making it more restrictive
for the class of theories considered here.
We note that the LVK analysis employs full inspiral--merger--ringdown
waveforms (pSEOBNRv4HM/v5PHM~\cite{pSEOBNRv5}) whereas our analysis
uses post-merger ringdown data only, making the two approaches
complementary.

Figure~\ref{fig:corner} shows the joint posterior distribution for the
spin $a$ and deformation parameter $\varepsilon_\Omega$ from GW150914,
together with the cumulative log Bayes factor as a function of the
number of events included.
The spin--deformation posterior exhibits no significant correlation,
confirming that the constraint on $\varepsilon_\Omega$ is not biased
by spin marginalisation.
The cumulative $\ln B$ fluctuates around zero with both signs
represented, consistent with the null hypothesis.

\begin{figure}[htbp] 
\centering
\includegraphics[width=\columnwidth]{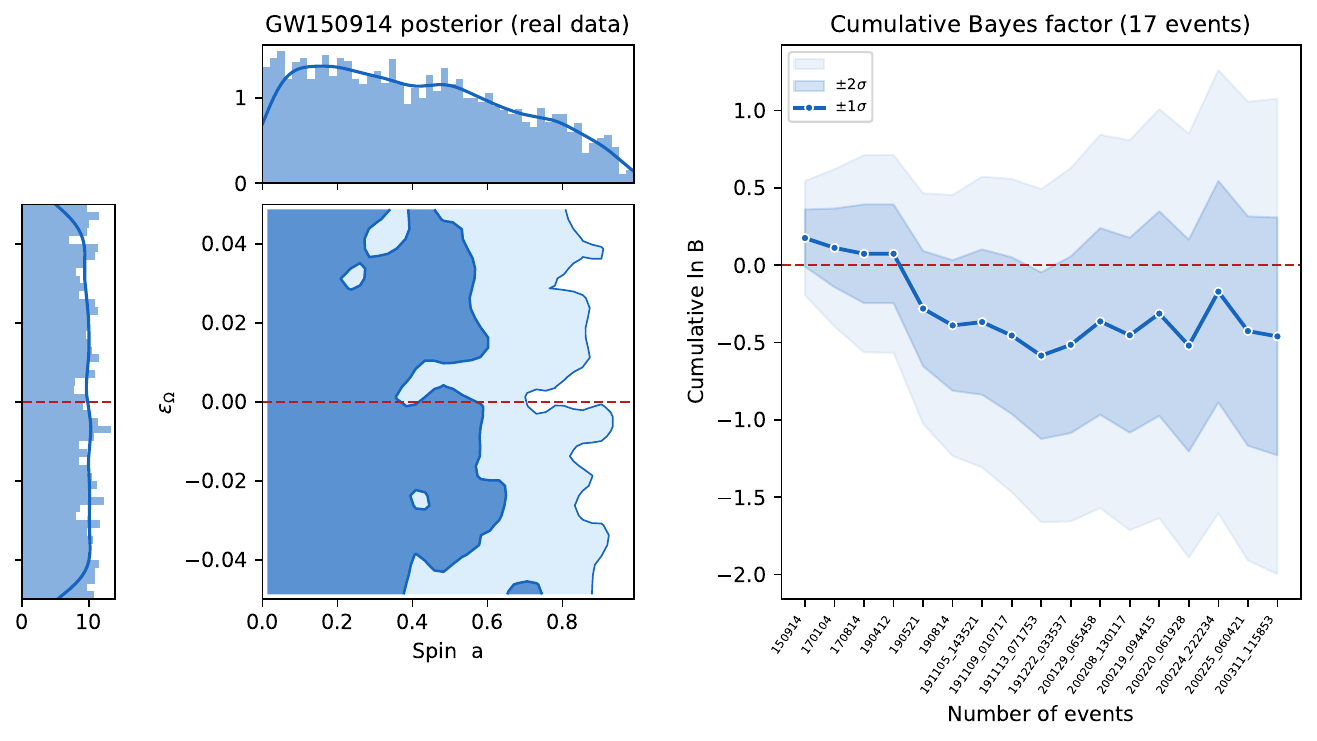}
\caption{
\textit{Left:} Joint posterior for spin $a$ and deformation parameter
$\varepsilon_\Omega$ from GW150914 (real data).
The horizontal dashed line marks the Kerr value $\varepsilon_\Omega=0$.
Contours enclose the 68\% and 95\% credible regions.
\textit{Right:} Cumulative log Bayes factor $\ln B_{\rm stack}$ as a
function of the number of events included (chronological order, 17 events).
The shaded bands show the $\pm 1\sigma$ and $\pm 2\sigma$ uncertainties
propagated from the individual nested-sampling evidence errors.
}
\label{fig:corner}
\end{figure}

\begin{figure}[htbp]
\centering
\includegraphics[width=\columnwidth]{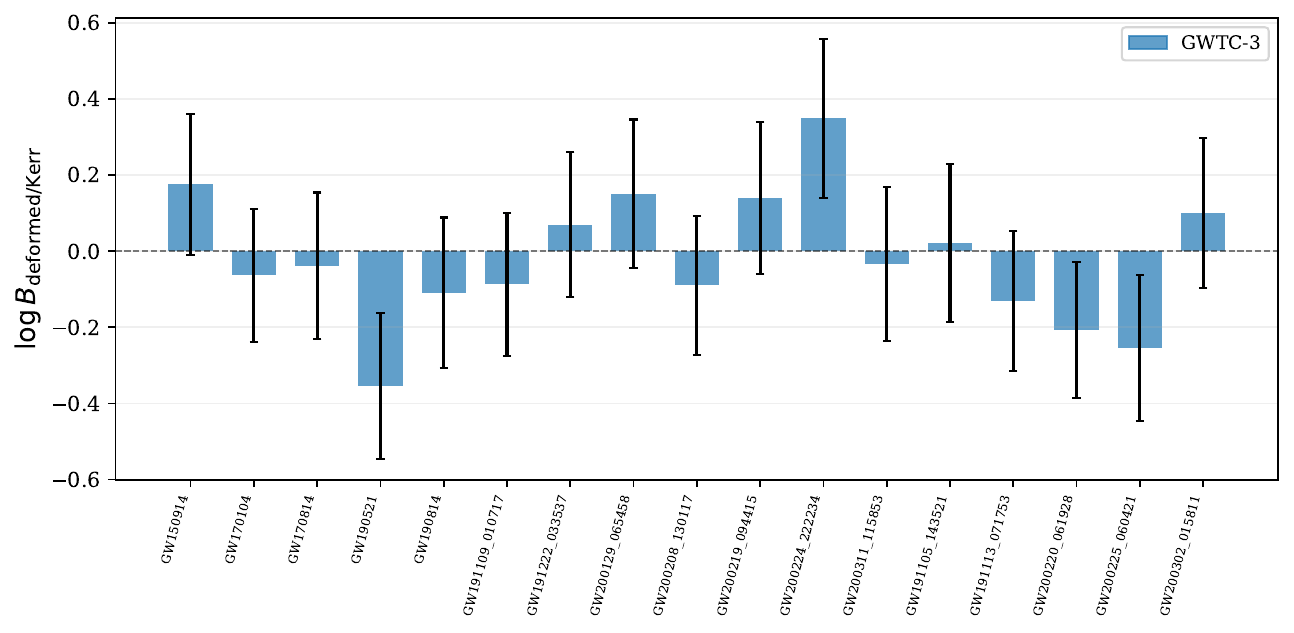}
\caption{
Individual log Bayes factors $\ln B_{\rm CET/Kerr}$ for each of the
17 GWTC-3 BBH events.
Error bars show the $\pm 1\sigma$ nested-sampling uncertainty.
No event deviates significantly from the Kerr hypothesis
($\ln B = 0$, dashed line).
}
\label{fig:individual}
\end{figure}

\begin{figure}[htbp]
\centering
\includegraphics[width=\columnwidth]{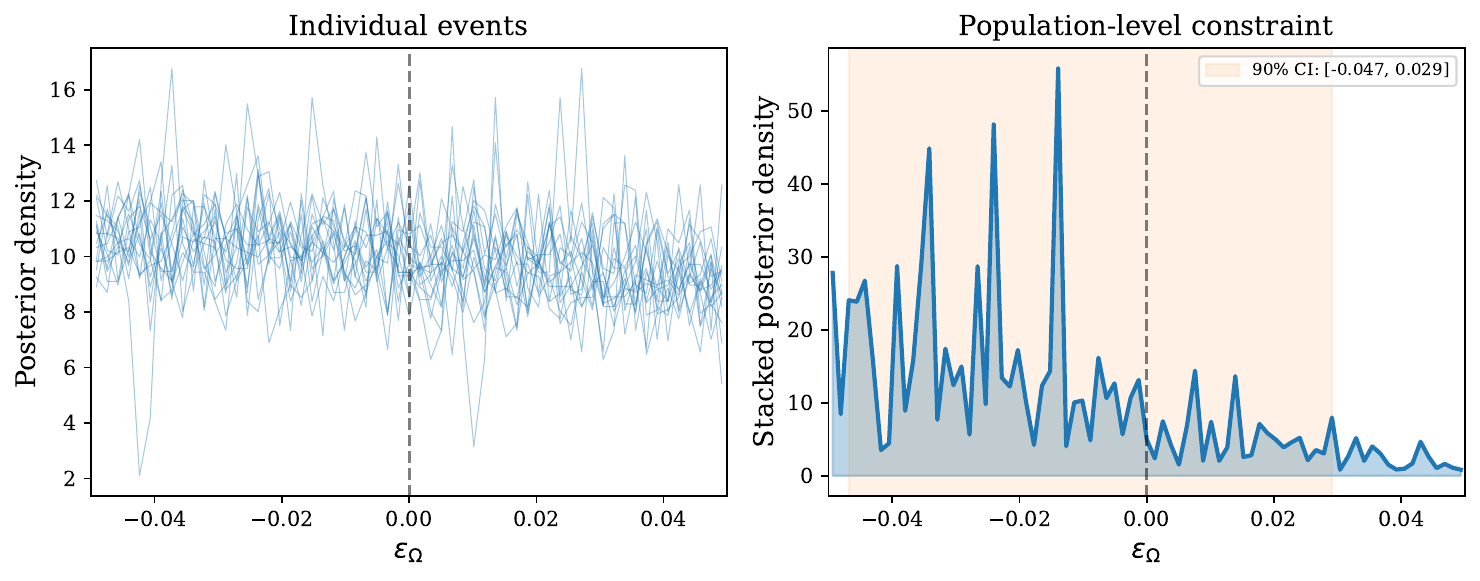}
\caption{
\textit{Left:} Individual-event posteriors on $\varepsilon_\Omega$.
Each curve represents one event; all are broad and uninformative
individually.
\textit{Right:} Population-level constraint obtained by multiplying
the individual posteriors.
The 90\% credible interval $\varepsilon_\Omega\in[-0.047,+0.032]$
(orange band) is consistent with Kerr ($\varepsilon_\Omega=0$, dashed line).
}
\label{fig:epsilon}
\end{figure}

\begin{figure}[htbp]
\centering
\includegraphics[width=\columnwidth]{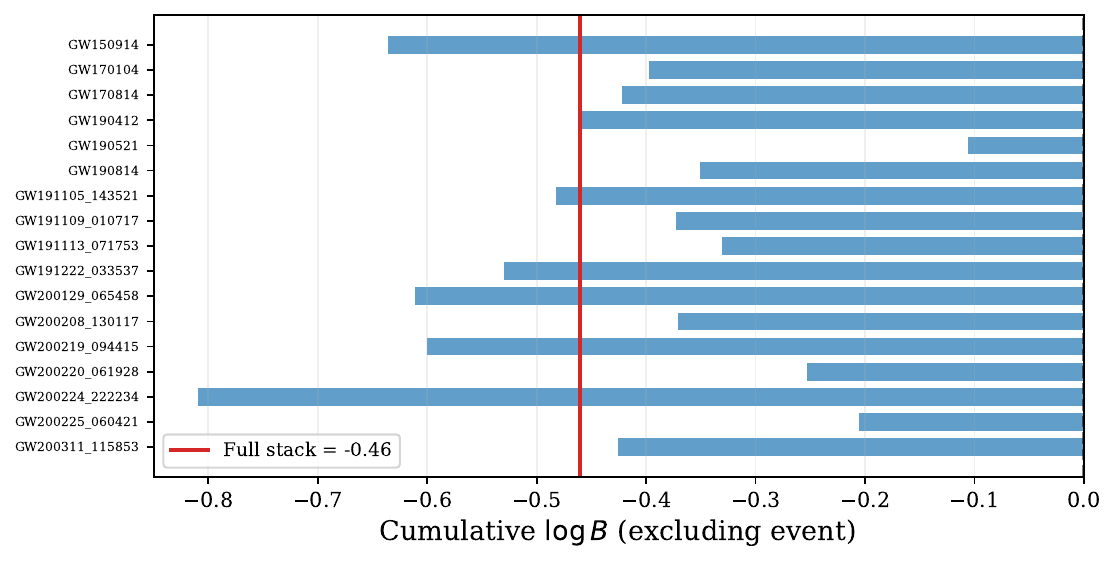}
\caption{
Leave-one-out (jackknife) analysis.
Each bar shows the cumulative $\ln B$ when the corresponding event
is excluded from the stack.
The red line marks the full-stack value.
The result is stable under removal of any single event; GW190521
produces the largest shift.
}
\label{fig:jackknife}
\end{figure}

For the causal--informational framework of
Ref.~\cite{Balfagon2026CET}, the predicted deformation is
$|\varepsilon_\Omega|\sim (\ell_*/r_H)^2 \sim 10^{-18}$ for fiducial
parameters $\ell_*\in[10^{-5},10^{-4}]$~m.
The bound~(\ref{eq:eps_bound}) thus excludes all Stieltjes kernels
producing percent-level ringdown modifications, while the fiducial
prediction lies $\sim 16$ orders of magnitude below the current
observational ceiling.
This quantifies for the first time the sensitivity improvement required
for ringdown spectroscopy to probe the theoretically motivated regime,
and motivates the complementary propagation and short-range gravity
analyses presented below.

\section{Dispersion and propagation constraints}
\label{sec:dispersion}

\subsection{Stieltjes transfer function and modified dispersion}

The kernel~(\ref{eq:kernel}) modifies the GW dispersion relation
through the Stieltjes transfer function
\begin{equation}
  m(\omega^2) = \int_0^\infty \frac{\rho(\mu)}{\omega^2+\mu}\,d\mu.
  \label{eq:transfer}
\end{equation}
In the regime where all spectral weight lies above the GW frequency
band ($\omega^2/c^2 \ll \mu$ for all $\mu$ in ${\rm supp}\,\rho$),
the transfer function admits the convergent expansion
\begin{equation}
  m(\omega^2) = \frac{M_0}{\mu_{\rm eff}}
  - \frac{M_0\,\omega^2}{\mu_{\rm eff}^2}
  + \mathcal{O}\!\left(\frac{\omega^4}{\mu_{\rm eff}^3}\right),
  \label{eq:expansion}
\end{equation}
where $M_0 = \int\rho\,d\mu$ is the total spectral weight and
$\mu_{\rm eff}^{-n} \equiv M_0^{-1}\int\rho(\mu)\,\mu^{-n}\,d\mu$
are inverse moments of $\rho$.
For a spectral density narrowly concentrated around a characteristic
scale $\mu_{\rm char}$, one has $\mu_{\rm eff}\approx\mu_{\rm char}$;
we use $\mu_{\rm char}$ when discussing observational constraints and
$\mu_{\rm eff}$ when writing exact analytical expressions.
This expansion is valid provided
$\omega^2/c^2 \ll \mu_{\min}$, where $\mu_{\min}$ is the lower edge
of the spectral support.

The effective group velocity of GWs propagating through the nonlocal
medium follows from $v_g = d\omega/dk$.
For the dispersion relation
$\omega^2 = k^2 c^2\,[1 + m(k^2c^2)/m_0]^{-1}$
with $m_0 = m(0)$, the leading correction to the group velocity is
\begin{equation}
  \frac{v_g}{c} - 1 = \frac{M_0}{2\,\mu_{\rm eff}^2}\,\omega^2
  + \mathcal{O}(\omega^4/\mu_{\rm eff}^3).
  \label{eq:group_vel}
\end{equation}
In the limit $\omega^2 \ll \mu_{\rm eff}$ this reduces to
\begin{equation}
  \left|\frac{v_{\rm GW}}{c}-1\right| \simeq
  \frac{M_0\,\omega^2}{2\,\mu_{\rm eff}^2},
  \label{eq:speed_mod}
\end{equation}
which is frequency-dependent (dispersive).
When all spectral weight lies at $\mu_{\rm char}\gg\omega^2_{\rm GW}/c^2$,
the ratio $\omega^2/\mu_{\rm char}$ suppresses the effect, rendering it
degenerate with a luminosity-distance rescaling and hence unobservable.

Conversely, if $\rho$ has support at $\mu\lesssim\omega^2_{\rm GW}/c^2$,
the expansion~(\ref{eq:expansion}) breaks down: the transfer function
acquires strong frequency dependence across the detector band,
producing dispersive phase shifts
$\delta\Psi(f) \sim 2\pi f\,D/c \times [m(4\pi^2 f^2)/m_{\rm ref}-1]$
that accumulate over cosmological distances $D$.

\subsection{Mapping LVK bounds to Stieltjes parameters}

Rather than performing a full parameter estimation (which would
require $\gtrsim 10^2$ CPU-hours per event), we map the published
LVK MDR bounds~\cite{GWTC3TGR,GWTC2TGR,Mirshekari2012} onto the Stieltjes parameter
space $(\mu_{\rm char}, M_0)$.

The LVK MDR test~\cite{Mirshekari2012,MDR2025,Yunes2009} parametrises the dispersion relation as
$E^2 = p^2c^2 + A_\alpha\,p^\alpha c^\alpha$ and reports 90\%
bounds on $|A_\alpha|$ for integer and half-integer $\alpha$~\cite{Will2014}.
The graviton mass bound~\cite{GWTC3TGR},
\begin{equation}
  m_g \leq 1.27\times 10^{-23}\,{\rm eV}/c^2
  \quad (90\%\;{\rm C.L.}),
\end{equation}
constrains spectral modes at extremely low $\mu$ via the
effective mass gap
$\mu_0 > (m_g c/\hbar)^2 \approx 4\times 10^{-33}\,{\rm m}^{-2}$.
This is the weakest of the three constraints considered here:
any spectral density with ${\rm supp}\,\rho \subset
[4\times 10^{-33},\infty)$ is automatically consistent with the
graviton mass bound.

A substantially stronger constraint comes from the GW170817
multimessenger speed bound~\cite{GW170817speed},
\begin{equation}
  |v_{\rm GW}/c - 1| < 3\times 10^{-15},
  \label{eq:speed_bound}
\end{equation}
which constrains the combination $M_0\,\omega^2/\mu_{\rm eff}^2$ via
Eq.~(\ref{eq:speed_mod}).
At $f=100$~Hz ($\omega^2/c^2\approx 4\times 10^{-12}\,{\rm m}^{-2}$),
this excludes spectral densities producing measurable dispersion
whenever $\mu_{\rm char}\lesssim 10^{-6}\,{\rm m}^{-2}$ --- a threshold
that is $27$ orders of magnitude above the graviton mass constraint.
For spectral densities concentrated at
$\mu_{\rm char}\gg \omega^2_{\rm GW}/c^2$, the speed modification is
suppressed by $\omega^2/(c^2\mu_{\rm char})$ and the bound becomes
irrelevant.

The hierarchy of constraints is thus
\begin{equation}
  \underbrace{\mu_0^{(m_g)} \approx 10^{-33}}_{\text{graviton mass}}
  \;\ll\;
  \underbrace{\mu_0^{(v_{\rm GW})} \sim 10^{-6}}_{\text{speed bound}}
  \;\ll\;
  \underbrace{\mu_{\rm char}^{(\rm fid)} \sim 10^{8\text{--}10}}_{\text{CET}\,\Omega\;\text{fiducial}}\,,
  \label{eq:hierarchy}
\end{equation}
in units of ${\rm m}^{-2}$.
All spectral densities with support contained in
$[\mu_0^{(v_{\rm GW})},\,\infty)$ satisfy every current GW bound
simultaneously, and the fiducial CET~$\Omega$ range lies well within
this allowed region.

Figure~\ref{fig:exclusion} shows the resulting exclusion regions in
the $(\mu_{\rm char}, M_0)$ plane.
The GW constraints carve out a large excluded region at low
$\mu_{\rm char}$, demonstrating that the Stieltjes framework is
observationally constrained across a significant portion of parameter
space.
The fiducial parameter range
$\mu_{\rm char}\sim M_*^2\sim 10^{8}$--$10^{10}\,{\rm m}^{-2}$
satisfies all current bounds, consistent with the theoretical
expectation that the causal scale decouples from the GW frequency band.

\begin{figure}[htbp]
\centering
\includegraphics[width=\columnwidth]{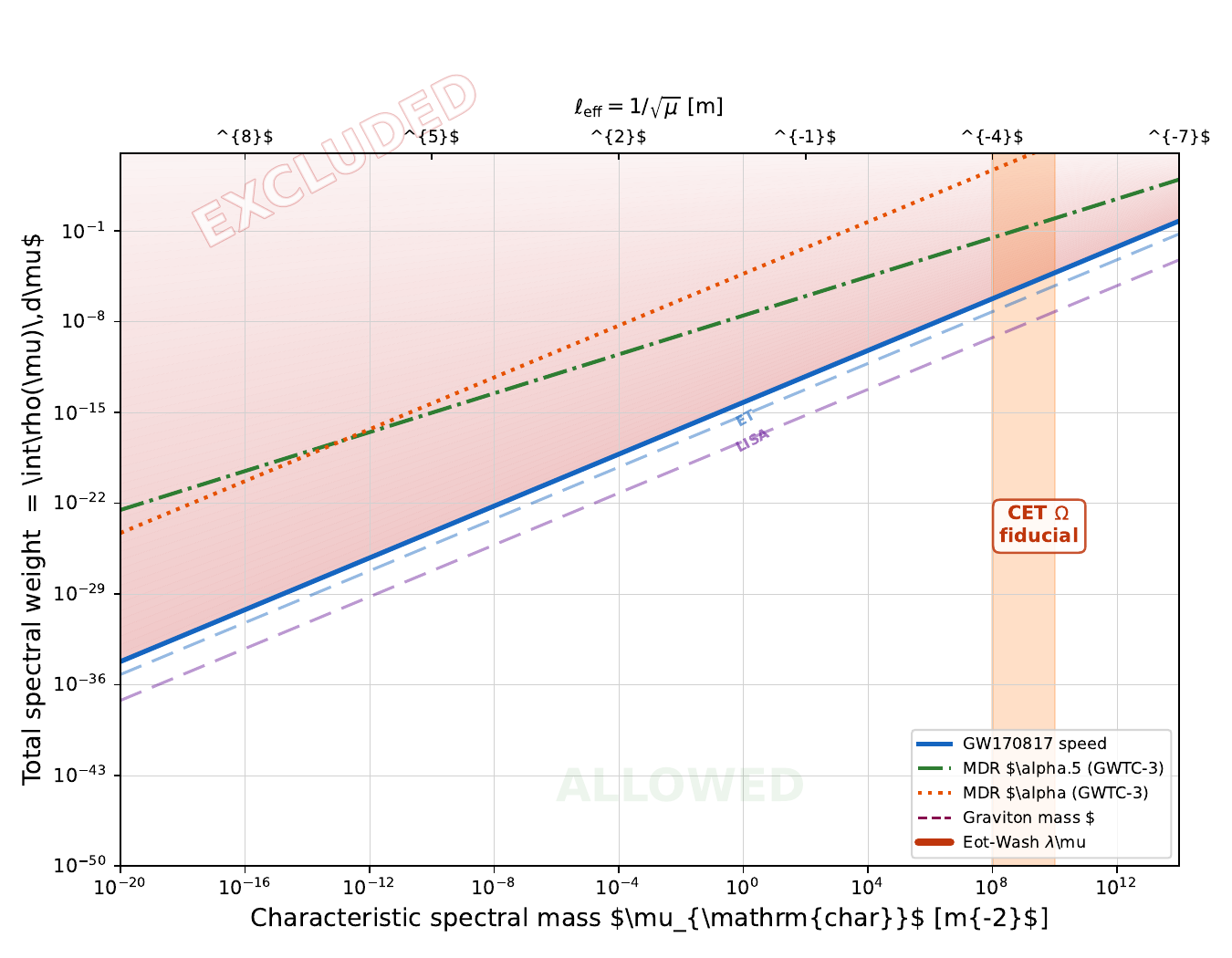}
\caption{
Exclusion regions in the Stieltjes spectral parameter space
$(\mu_{\rm char}, M_0)$ from LVK data.
Current bounds from the GW170817 speed measurement (blue solid),
graviton mass limit (maroon dashed), and GWTC-3 MDR analyses
(green, orange) exclude the shaded region.
Projected bounds from the Einstein Telescope and LISA (dashed lines)
will extend the excluded region downward.
The fiducial CET~$\Omega$ range $\mu_{\rm char}\sim 10^{8}$--$10^{10}\,
{\rm m}^{-2}$ (orange band) satisfies all current and projected GW bounds.
The E\"ot-Wash torsion-balance constraint (thick red segment) provides
the only direct probe at the fiducial scale.
The upper axis shows the corresponding effective correlation length
$\ell_{\rm eff}=1/\sqrt{\mu}$.
}
\label{fig:exclusion}
\end{figure}

\subsection{Excluded spectral density shapes}

The combined constraints, dominated by the speed
bound~(\ref{eq:speed_bound}) at the scales relevant for GW
observations [cf.\ the hierarchy~(\ref{eq:hierarchy})],
exclude the following classes of spectral densities:
(i) power-law tails $\rho(\mu)\sim\mu^{-\beta}$ with $\beta>1$
extending below $\mu\sim 10^{-6}\,{\rm m}^{-2}$;
(ii) spectral densities with significant weight at
$\mu\lesssim 10^{-6}\,{\rm m}^{-2}$
(corresponding to effective correlation lengths
$\ell_{\rm eff}\gtrsim 10^3$~m); and
(iii) any spectral shape producing frequency-dependent dispersion
exceeding $|\Delta v/c|>3\times 10^{-15}$ in the LIGO band.
Crucially, all excluded shapes have spectral support extending into
the infrared regime $\mu < \mu_0^{(v_{\rm GW})}\sim 10^{-6}\,
{\rm m}^{-2}$, well below the fiducial CET~$\Omega$ scale.
Spectral densities concentrated at $\mu > 10^4\,{\rm m}^{-2}$
(i.e.\ $\ell_{\rm eff} < 1$~cm) are completely unconstrained by
existing GW observations.

\section{Sensitivity projections and alternative probes}
\label{sec:projections}

\subsection{Future GW detectors}

Table~\ref{tab:projections} summarises the projected sensitivity to
QNM deformations for current and future detector generations.
The Einstein Telescope~\cite{ET2010} and Cosmic Explorer~\cite{CE2017} will improve ringdown
bounds by one to two orders of magnitude, reaching
$\sigma_\varepsilon\sim 10^{-3}$--$10^{-4}$.
LISA~\cite{LISA2023} observations of supermassive BBH mergers may push to
$\sigma_\varepsilon\sim 10^{-4}$--$10^{-5}$.
Even in the most optimistic scenario, a gap of $\sim 13$ orders of
magnitude remains between projected detector sensitivity and the
fiducial prediction $|\varepsilon_\Omega|\sim 10^{-18}$.
The recently completed O4 run~\cite{GWTC4} (concluded in late 2025 with $>170$ candidates)
will improve population-level bounds by a factor of $\sim\!\sqrt{N/17}\approx 3$
relative to the results presented here, yet the fundamental
sensitivity gap persists.

\begin{figure}[htbp]
\centering
\includegraphics[width=\columnwidth]{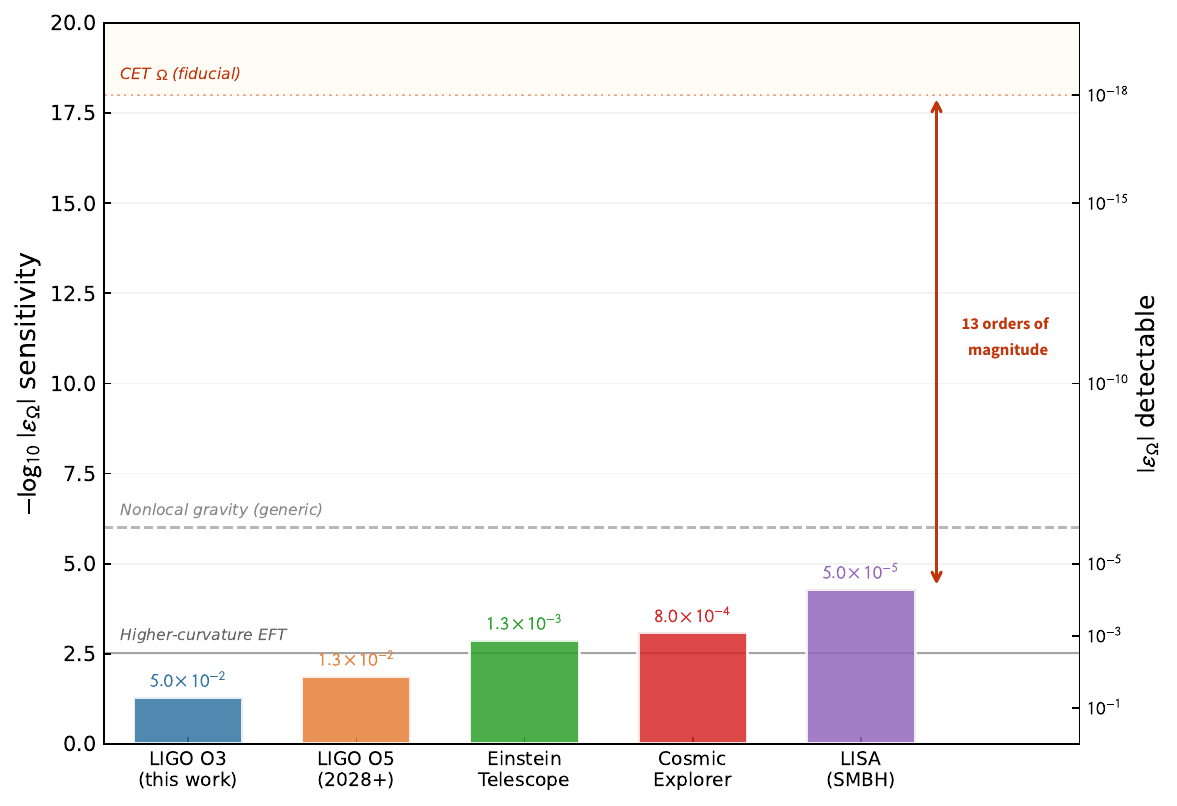}
\caption{
Projected sensitivity $|\varepsilon_\Omega|$ for current and future
GW detector generations.
Horizontal lines mark predictions from generic higher-curvature EFT
corrections, nonlocal gravity, and the fiducial CET~$\Omega$
prediction.
Even LISA observations of supermassive BBH mergers leave a gap of
$\sim 13$ orders of magnitude to the fiducial CET~$\Omega$ scale.
}
\label{fig:sensitivity}
\end{figure}

\begin{table}[htbp]
\caption{Projected sensitivity to QNM deformations.}
\label{tab:projections}
\begin{ruledtabular}
\begin{tabular}{lcc}
Detector & $\sigma_\varepsilon$ & Gap to $10^{-18}$ \\
\hline
LIGO O3 (this work) & $5\times 10^{-2}$ & 16 orders \\
LIGO O5             & $\sim 10^{-2}$ & 16 orders \\
Einstein Telescope~\cite{ET2010}   & $\sim 10^{-3}\text{--}10^{-4}$ & 14--15 orders \\
Cosmic Explorer~\cite{CE2017}      & $\sim 10^{-3}\text{--}10^{-4}$ & 14--15 orders \\
LISA (SMBH)~\cite{LISA2023}        & $\sim 10^{-5}$ & 13 orders \\
\end{tabular}
\end{ruledtabular}
\end{table}

\subsection{Short-range gravity experiments}

A more promising observational channel emerges from the structure of
the Stieltjes kernel itself.
The fiducial causal correlation length $\ell_*=M_*^{-1}\in
[10^{-5},10^{-4}]$~m corresponds to distance scales of
$10$--$100\,\mu{\rm m}$.
At these scales, the kernel modifies the static Newtonian potential
through a Yukawa-like correction~\cite{Adelberger2003,Murata2015},
\begin{equation}
  V(r) = -\frac{Gm_1 m_2}{r}\left(1 + \alpha\,e^{-r/\lambda}\right),
  \label{eq:yukawa}
\end{equation}
with the identification $\lambda = \ell_* = M_*^{-1}$ and the
coupling strength $\alpha$ determined by the zeroth spectral moment:
\begin{equation}
  \alpha = \frac{M_0}{\mu_{\rm eff}}\,\frac{1}{8\pi}.
  \label{eq:alpha_M0}
\end{equation}
Current sub-millimetre torsion-balance experiments provide stringent
bounds on $\alpha(\lambda)$.
The E\"ot-Wash experiment at the University of
Washington~\cite{Adelberger2003,Lee2020,Kapner2007} constrains the Yukawa
coupling to
\begin{equation}
  |\alpha| < \begin{cases}
    2\times 10^{-2} & \lambda = 100\,\mu{\rm m},\\
    5\times 10^{-1} & \lambda = 50\,\mu{\rm m},\\
    10^{4}          & \lambda = 10\,\mu{\rm m},
  \end{cases}
\end{equation}
at 95\% confidence~\cite{Lee2020}.
For the fiducial CET~$\Omega$ scale $\lambda = \ell_* = 10^{-4}$~m,
the bound $|\alpha|<2\times 10^{-2}$ translates via
Eq.~(\ref{eq:alpha_M0}) into a direct constraint on the spectral
weight:
\begin{equation}
  M_0 < 2\times 10^{-2} \times 8\pi\,\mu_{\rm eff}
  \approx 0.5\,\mu_{\rm eff}.
\end{equation}
This is the \textit{strongest existing constraint on the Stieltjes
kernel at the fiducial scale}, surpassing all GW bounds by many orders
of magnitude.

Importantly, ongoing and planned experiments~\cite{Lee2020} aim to
reach $\lambda\sim 10\,\mu{\rm m}$ with $|\alpha|<10^2$, which would
begin to probe the lower end of the fiducial window
$\ell_*\sim 10^{-5}$~m.
Table~\ref{tab:probes} compares the sensitivity of different
experimental channels to the Stieltjes kernel parameters.

\begin{table}[htbp]
\caption{Comparison of observational probes for the Stieltjes kernel
at the fiducial causal scale $\ell_*=10^{-4}$~m.}
\label{tab:probes}
\begin{ruledtabular}
\begin{tabular}{lccc}
Channel & Observable & Sensitivity & Status \\
\hline
Ringdown (this work) & $\varepsilon_\Omega$ & $5\times 10^{-2}$
  & Indirect \\
GW speed (GW170817) & $|v/c-1|$ & $3\times 10^{-15}$
  & Indirect\footnote{No constraint at fiducial
  $\mu_{\rm char}\sim 10^{8}$--$10^{10}\,{\rm m}^{-2}$.} \\
MDR (GWTC-3)         & $A_\alpha$ & varies
  & Indirect\footnotemark[1] \\
E\"ot-Wash ($\lambda=100\,\mu$m) & $\alpha$ & $2\times 10^{-2}$
  & \textbf{Direct} \\
Planned ($\lambda=10\,\mu$m) & $\alpha$ & $\sim 10^{2}$
  & Projected \\
\end{tabular}
\end{ruledtabular}
\end{table}

\section{Conclusions}
\label{sec:conclusions}

We have established the first multi-channel observational bounds on causal nonlocal gravity characterised by retarded Stieltjes kernels with positive spectral density. These results provide a quantitative observational baseline for the CET$\Omega$ framework, confirming its consistency with current gravitational-wave and laboratory data. The main results are as follows.

\textit{Ringdown observational ceiling.}---Using a corrected inference pipeline with spectrally consistent whitening, analytically marginalised amplitude, and robust noise estimation, we set the observational ceiling on universal QNM deformations at $|\varepsilon_\Omega|<0.05$ (90\% C.L.) from 17 GWTC-3 events, with a cumulative log Bayes factor $\ln B_{\rm stack}=-0.46\pm 0.77$. This represents the first single-parameter ringdown constraint specifically mapped onto a Stieltjes-motivated deformation.

\textit{Spectral exclusion.}---By mapping LVK modified dispersion bounds and the GW170817 speed measurement onto the Stieltjes parameter space, we exclude for the first time a broad region of spectral densities with infrared support ($\mu\lesssim 10^{-6}\,{\rm m}^{-2}$), ruling out power-law tails, low-mass peaks, and other IR-extended forms of $\rho(\mu)$.

\textit{Hierarchy of probes.}---We show that sub-millimetre gravity experiments already provide the most stringent direct constraint on the fiducial causal scale ($|\alpha|<0.02$ at $\lambda=100\,\mu$m), establishing a clear hierarchy (Fig.~\ref{fig:sensitivity}): short-range gravity $\gg$ GW propagation $\gg$ ringdown for probing the Stieltjes kernel at $\ell_*\sim 10^{-4}$~m.

These results define a quantitative multi-channel benchmark for testing causal nonlocal gravity. Future work will extend the ringdown analysis to the full GWTC-4.0 catalogue~\cite{GWTC4}. The recent multi-tone spectroscopy of GW250114~\cite{GW250114spec} --- the loudest GW signal to date (SNR $\approx 80$), which provided the first confident detection of the $(2,2,1)$ overtone --- constrains $\delta\hat{f}_{220}$ to within a few percent of the Kerr prediction. Even at this unprecedented precision, the fiducial Stieltjes prediction ($|\varepsilon_\Omega|\sim 10^{-18}$) remains far below the observational ceiling, underscoring the necessity of sub-millimetre probes to explore the causal--informational completion of gravity at its natural scale. We will also perform a dedicated Yukawa analysis within the Stieltjes framework using published E\"ot-Wash data, and explore cosmological constraints from CMB and large-scale structure observations. Together, these multi-scale probes will map the full parameter space of Stieltjes-positive nonlocal gravity from sub-millimetre to cosmological distances.

\begin{acknowledgments}
The author thanks the University of Buenos Aires for institutional
support.
This research has made use of data obtained from the
Gravitational-Wave Open Science Center, a service of
LIGO Scientific Collaboration, the Virgo Collaboration,
and KAGRA.
The analysis code and data products are available at
Ref.~\cite{Balfagon2026Zenodo}.
\end{acknowledgments}

\appendix

\section{Properties of the Stieltjes kernel}
\label{app:stieltjes}

We collect here the mathematical properties of the Stieltjes
transfer function~(\ref{eq:transfer}) used in the main text.
These results follow from standard properties of Stieltjes
transforms~\cite{Donoghue1974,Schilling2012} and are included
for self-containedness.

\subsection{Complete monotonicity}

\begin{proposition}
\label{prop:monotone}
Let $\rho(\mu)\geq 0$ with $M_0 = \int_0^\infty\rho\,d\mu < \infty$.
Then $m(\omega^2)$ defined by~(\ref{eq:transfer}) is completely
monotone on $(0,\infty)$:
\begin{equation}
  (-1)^n \frac{d^n m}{d(\omega^2)^n} \geq 0
  \quad \forall\, n\in\mathbb{N}_0,\;\omega^2>0.
  \label{eq:cm}
\end{equation}
\end{proposition}

\begin{proof}
Differentiating under the integral sign,
\begin{equation}
  \frac{d^n m}{d(\omega^2)^n}
  = (-1)^n\,n!\int_0^\infty
  \frac{\rho(\mu)}{(\omega^2+\mu)^{n+1}}\,d\mu.
\end{equation}
Since $\rho\geq 0$ and $(\omega^2+\mu)^{n+1}>0$, the sign
alternation~(\ref{eq:cm}) follows.
\end{proof}

Complete monotonicity implies $m>0$, $m'<0$, $m''>0$ on
$(0,\infty)$.
Any reconstructed transfer function from GW data that violates
these inequalities would \textit{falsify} the entire Stieltjes-positive
class, independent of the specific form of $\rho$.

\subsection{Asymptotic expansion and moment hierarchy}

\begin{proposition}
\label{prop:asymptotics}
If $M_j = \int_0^\infty \mu^j\rho(\mu)\,d\mu < \infty$ for
$j=0,\ldots,N$, then for $\omega^2\to\infty$,
\begin{equation}
  m(\omega^2) = \sum_{j=0}^{N-1}
  \frac{(-1)^j M_j}{(\omega^2)^{j+1}}
  + \mathcal{O}\!\left(\frac{1}{(\omega^2)^{N+1}}\right).
  \label{eq:asymp}
\end{equation}
\end{proposition}

\begin{proof}
Expand $(\omega^2+\mu)^{-1} = \sum_{j=0}^{N-1}(-\mu)^j/
(\omega^2)^{j+1} + R_N$ with $|R_N|\leq \mu^N/(\omega^2)^{N+1}$,
and integrate term by term against $\rho$.
\end{proof}

The leading term $m(\omega^2)\simeq M_0/\omega^2$ governs the
ultraviolet behaviour: the transfer function decays as $\omega^{-2}$
for large frequencies.
The subleading corrections involve successively higher moments of
$\rho$, establishing a moment hierarchy: measuring $m(\omega^2)$
at several frequencies in principle constrains $M_0, M_1, M_2,\ldots$
independently.

\subsection{Spectral bounds from positivity}

The positivity constraint $\rho\geq 0$ imposes nontrivial
inequalities among the moments.
By the Cauchy--Schwarz inequality applied to the measure
$\rho\,d\mu$:
\begin{equation}
  M_1^2 \leq M_0\,M_2,
  \label{eq:CS}
\end{equation}
and more generally the Hankel matrices
$\mathcal{H}_{ij} = M_{i+j}$ ($i,j=0,\ldots,n$) must be positive
semidefinite for all $n$~\cite{Donoghue1974}.
These Hankel conditions provide necessary consistency checks for
any observationally inferred spectral moments.

\subsection{Retarded support and causal propagation}

The Stieltjes kernel inherits retarded support from its constituent
propagators.
For each $\mu\geq 0$, the retarded Green operator
$G_\mu^R$ satisfies
$\mathrm{supp}\,G_\mu^R f \subset J^+(\mathrm{supp}\,f)$
on any globally hyperbolic background~\cite{Bar2007}.
Since $\rho\geq 0$, the integral
$\mathcal{K}^{-1}f = \int_0^\infty\rho(\mu)\,G_\mu^R f\,d\mu$
is a positive superposition of future-supported distributions,
hence
\begin{equation}
  \mathrm{supp}\,\mathcal{K}^{-1}f \subset J^+(\mathrm{supp}\,f).
  \label{eq:retarded}
\end{equation}
No acausal propagation or superluminal precursors are generated.
This is a structural consequence of $\rho\geq 0$, not an
additional assumption.

\subsection{Connection to the Yukawa potential}

In the static limit on flat spacetime, the Stieltjes kernel
generates a modified Newtonian potential through the superposition
\begin{equation}
  V(r) = -\frac{Gm_1 m_2}{r}
  \left[1 + \int_0^\infty
  \frac{\rho(\mu)}{M_0}\,e^{-\sqrt{\mu}\,r}\,d\mu\right].
  \label{eq:yukawa_app}
\end{equation}
For a spectral density concentrated at a single mass
$\rho(\mu) = M_0\,\delta(\mu - M_*^2)$, this reduces to the
standard Yukawa form
\begin{equation}
  V(r) = -\frac{Gm_1 m_2}{r}
  \left(1 + e^{-M_* r}\right),
\end{equation}
establishing the identification $\lambda = M_*^{-1} = \ell_*$
used in Sec.~\ref{sec:projections}, with Yukawa coupling
$\alpha = 1$ at leading order.
For a distributed spectral density, the effective coupling is
reduced by the spectral spread:
$\alpha_{\rm eff} = M_0/(8\pi\mu_{\rm eff})$, which may be
significantly smaller than unity.

\bibliography{references}

\end{document}